\documentclass[12pt]{amsart}
\usepackage{amsmath,amscd,amssymb,latexsym}
\usepackage{graphicx}
\newtheorem{theorem}{Theorem}
\newtheorem{lemma}{Lemma}

\def\a{{\alpha}}
\def\b{{\beta}}

\def\d{{\delta}}
\def\e{{\epsilon}}

\begin{document}
\title{A 2-chain can interlock with an open $10$-chain}

\author{Bin Lu}
\address{Department of Mathematics \& Statistics\\California State University Sacramento\\6000 J Street\\ Sacramento, CA 95819}
\email{binlu@csus.edu}

\author{Joseph O'Rourke}
\address{Department of Computer Science\\ Smith College\\ Northampton, MA 01063}
\email{orourke@cs.smith.edu}

\author{Jianyuan K. Zhong}
\address{Department of Mathematics \& Statistics\\California State University Sacramento\\6000 J Street\\ Sacramento, CA 95819}
\email{kzhong@csus.edu}

\begin{abstract}
It is
an open problem, posed
in~\cite{SoCG},
to determine the minimal $k$ such that an open flexible $k$-chain can
interlock with a flexible $2$-chain. It was first established in~\cite{GLOSZ}
that there is an open $16$-chain in a trapezoid frame
that achieves interlocking. This was subsequently improved in~\cite{GLOZ} to establish interlocking between a 2-chain and
an open 11-chain. Here we improve that result once more,
establishing interlocking between a 2-chain and a 10-chain.
We present arguments that indicate that 10 is likely the minimum.
\end{abstract}
\maketitle

\section{Introduction}
An {\it open chain} is a linkage of
rigid bars (links, line segments, edges) connected at their joints (vertices, endpoints), which forms a simple, unclosed path.  A {\it folding} of a linkage is any  reconfiguration of the linkage obtained by moving the joints such that:
(1) The edges remain straight;
(2) The number and length of edges are preserved;
(3) The edges do not intersect or pass through one another.
When the joints of an open chain act as
universal joints, it is called a \emph{flexible open chain}. Often we drop the
prefixes ``flexible" and ``open" when understood from the context.
A chain with $k$ edges is called a {\it $k$-chain}. When emphasizing the edges in a chain, a $k$-chain is often referred as a $k$-link chain. A
collection of chains are {\it interlocked} if foldings cannot separate them.

Interlocking of open chains was studied in~\cite{DLOS,SoCG},
establishing a number of results regarding which collections of
chains can and cannot interlock.  An open problem posed in~\cite{SoCG} is: What is the
minimal number $k$ such that an open, flexible $k$-chain can
interlock with a flexible $2$-chain?  It was first established in~\cite{GLOSZ}
that there is an open $16$-chain in a trapezoid frame
that achieves interlocking. This was subsequently improved in~\cite{GLOZ} to establish interlocking between a 2-chain and
an open 11-chain. Here we improve that result once more,
establishing interlocking between a 2-chain and a 10-chain.
We present arguments that indicate that 10 is likely the minimum.

Here we summarize results from~\cite{DLOS} that we use in the sequel:
\begin{enumerate}
\item No collection of 2-chains can interlock.
\item Two open 3-chains cannot interlock, even with an additional collection of an arbitrary number of 2-chains.
\item A 2-chain cannot interlock with an open 4-chain. 
\item A flexible open 3-chain can interlock with a flexible open 4-chain.
\end{enumerate}

\noindent
A construction of a 3-chain interlocking a 4-chain, which we call a \emph{3/4-tangle}, is repeated below in Figure 1.

\begin{figure}[htbp]
\centering
\includegraphics[width=0.95\linewidth]{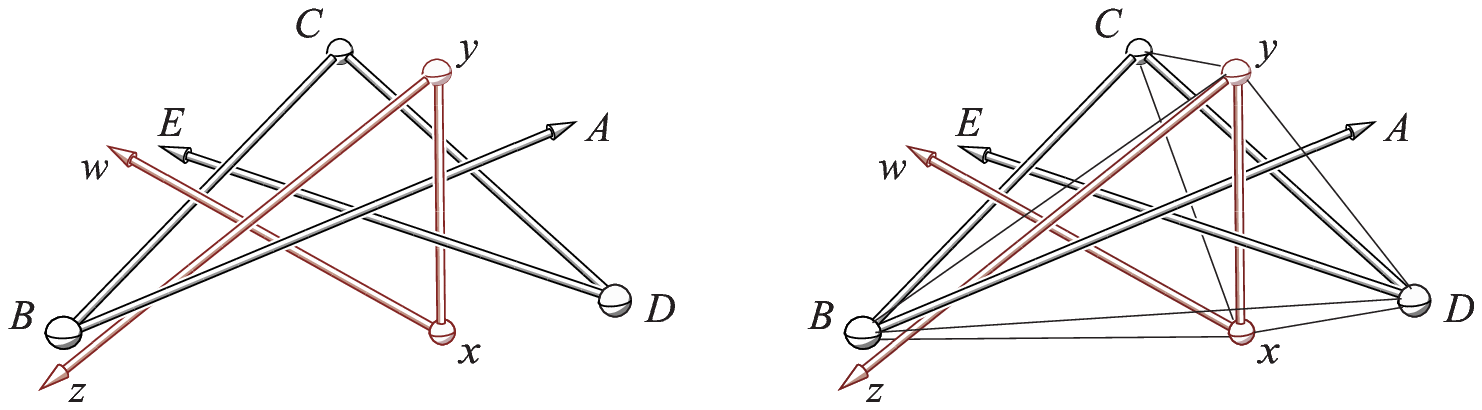}
\caption{Fig.~6 from \protect\cite{SoCG}.}
\label{link34}
\end{figure}
Some facts about the 3/4-tangle:
\begin{enumerate}
\item Theorem~11 of \cite{SoCG}: The convex hull
$CH(B,C,D,x,y)$ of the flexible
joints $B$, $C$, $D$, $x$, and $y$
does not change combinatorially as the chains are reconfigured.
\item Corollary 1 of \cite{GLOSZ}: Let $\e>0$ and $P$ be the midpoint
of $xy$, if $BC=CD=xy=\frac{1}{6}\e$, and end links
$AB=DE=xw=yz=\frac{1}{2}\e$, then all joints $B,C,D,x,y$ and endpoints $A, E, w, z$
stay inside the $\e$-ball centered at $P$.

\end{enumerate}

 This 3/4-tangle is crucial to the construction of the $11$-chain in \cite{GLOZ} and the construction of the $10$-chain here. In both constructions, the 11-chain and 10-chain are positioned to form a nearly rigid triangular frame which interlock the 2-chain. Note that there are other interlockings of a 3-chain with a 4-chain which could work, but
we follow \cite{GLOZ} in using the 3/4-tangle in Figure 1 to construct the 10-chain.

To make the 3/4 tangle part of a single open chain, there are four ways to connect the 3-chain with the 4-chain by adding additional link(s). In the discussion in~\cite{GLOZ} on ways to possibly further reduce the number of links in the 11-chain, two possibilities of allowing maximal sharing of the 3/4-tangle's 7 links were ruled out.
In this paper, we find that one of the two remaining possibilities can allow maximal sharing of two of the 3/4-tangle's 7 links in the triangle frame. This allows us to prove it indeed provides an open 10-chain interlocking with an open 2-chain.

The idea of proof and proof techniques are very similar to that in \cite{GLOZ},
and so we follow a similar proof structure.
The main difference in proof idea is that the triangle frame shares two of the 7 links of the 3/4-tangle. Therefore the triangle frame employs a total of $3+(5+1+1)=10$ links. We found it unnecessary to assume that the triangle frame is an isosceles triangle,
and so we alter Lemma 2 in \cite{GLOZ} accordingly.

\section{Idea of Proof}
The main idea of the proof, as in \cite{GLOZ}, is to build a
``rigid" triangular frame with small rings at its vertices
$(T_1, T_2, T_3)$, which could
interlock with a 2-chain, as shown in
Figure~\ref{Triangle.idea}(a) where $\triangle T_{1}T_{2}T_{3}$ is shaded. For then pulling away the vertex $v$ of the
2-chain from the triangle frame would diminish the
angle $\a$, and pushing $v$ up toward the triangle
would increase $\a$. But the only slack provided for $\a$ is that
determined by the diameter of the rings, as otherwise the triangular frame is rigid.

\begin{figure}[htbp]
\centering
\includegraphics[width=0.8\linewidth]{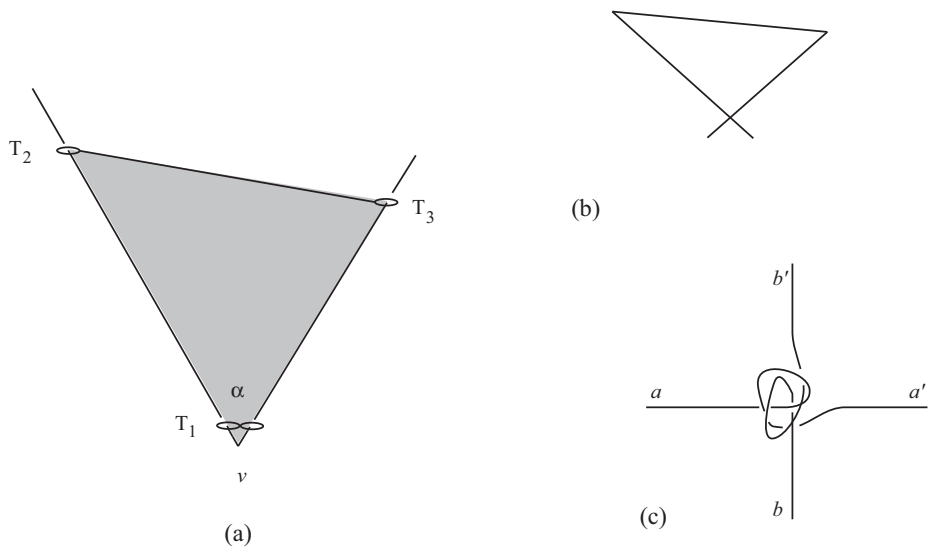}

\caption{(a) A rigid triangle with rings would interlock with a
2-chain; (b)~An open chain that simulates a rigid triangle;
(c)~Fixing a crossing of $aa'$ with $bb'$.}
\label{Triangle.idea}
\end{figure}


We can construct such a triangle using three links.  At each vertex, we take
one subchain $aa'$ and confine its crossing with another subchain
$bb'$ to within a small region of space. See
Figure~\ref{Triangle.idea}(c) for the idea. This pinning can be achieved by the ``$3/4$-tangle'' interlocking from~
\cite{SoCG} and ``jag loops.'' So the idea is to replace vertex $T_1$
with a small copy of the 3/4-tangle configuration. This can be
accomplished with $7$ links for a $3/4$-tangle, but
maximal sharing with both the incident incoming and outgoing triangle links reduces
the number of links needed.  We can achieve
confinement with $5$ links at the tangle near $T_1$.
At the other two
vertices of the triangle, this can be accomplished with
one extra link per vertex. Therefore, together with the $3$ links for the
main triangle skeleton, we use a total of $3+ (5+1+1) = 10$
links.

\section{A 2-chain can interlock an open 10-chain}
\subsection{A 2-chain can interlock an open 10-chain}
We take a 3/4-tangle whose joints
all stay within an $\e$-ball centered at
the midpoint of the middle link $xy$ of the 3-chain. 
Position the tangle
to represent the bottom vertex of a triangle, with the links arranged
as shown in Figure~\ref{10chain}. At each of the other two vertices, we add an extra link to make a \emph{jag loop}.

\begin{figure}[htbp]
\centering
\includegraphics[width=1.0\linewidth]{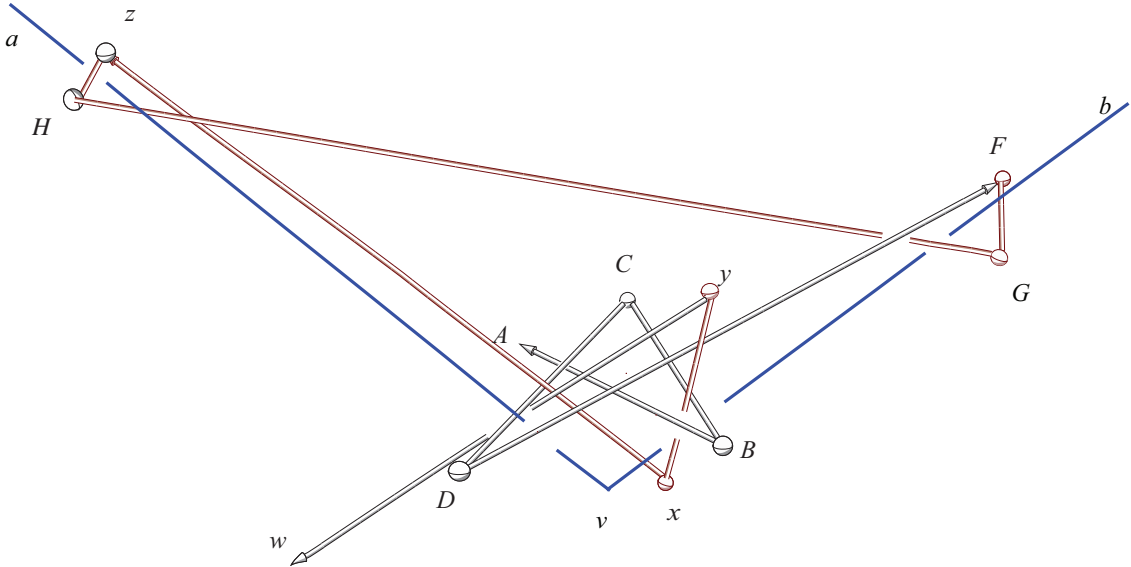}
\caption{An open 10-chain forming a nearly rigid triangle interlocking a 2-chain.}
\label{10chain}
\end{figure}

The configuration in Figure~\ref{10chain} is realizable,\footnote{
  We have constructed physical models of this configuration.}
that is,
there is more than enough flexibility in the design to ensure that
$va$ and $vb$ can indeed meet at the 2-chain vertex $v$. To see
that this is so, notice that the 3-chain and the 4-chain in the
3/4-tangle can be configured to lie in planes that are nearly orthogonal.
If we arrange the plane of the 2-chain at angles with respect to the two orthogonal planes within which the 3-chain and 4-chain lie,
then we can thread the 2-chain ($va$ following the link $xz$ and $vb$ following $DF$) as in Figure~\ref{10chain} to
pass through
one jag loop in the 4-chain and one loop made up by the end link $AB$ and the links of the 3-chain at vertex $x$,
at the same time as
weaving through the two jag loops at the base of the triangular
frame from above, as depicted in the figure.

\subsubsection{2-chain Through Triangle Jag Corners}
The jag loops at the other two corners can be assured to remain
in an $\e$-ball by making the extra link
lengths ($zH$ and $FG$) shorter than $\e$. We can conclude exactly as in \cite{GLOZ} that the link $va$ passes through an $\e$-ball at $z$ and cannot unlock with the jag. Note the link $vb$ passing through the jag loop at $F$ is a flip of the case $va$ passing through the jag loop at $z$. Thus we have that all corners of the triangle stay
within $\e$-balls, and both $va$ and $vb$ interlock the jag loops.
\begin{figure}[htbp]
\centering
\includegraphics[width=1.0\linewidth]{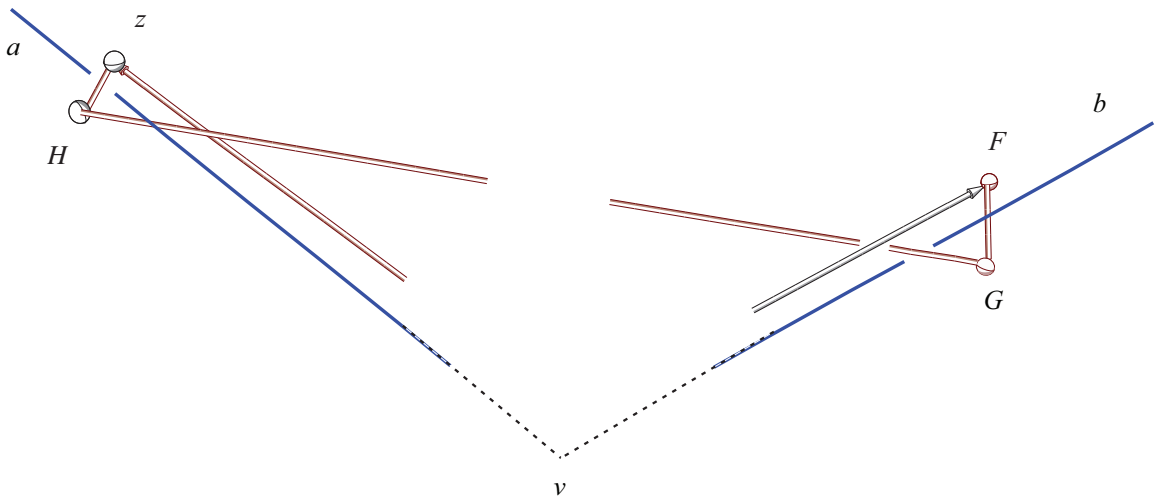}
\caption{ 1-link ``jags.''}
\label{1-link.jag}
\end{figure}

\subsubsection{The 2-chain Links Are Trapped by the 3/4-tangle }
We have exactly the same Lemma 1 as in \cite{GLOZ}, which we state again below. Our proof is slightly different due to the difference in construction.

\begin{lemma}  The vertex $v$ of the 2-chain cannot ``unweave''
from the 3/4-tangle.  Thus the links of the 2-chain are trapped by
the 3/4-tangle.
\label{Trapped.Lemma}
\end{lemma}
\begin{proof}

From the construction of the $10$-chain, link $va$ pierces $\triangle DCF$,  and links $vb$ and $va$ straddle
$yw$, $BC$ and $AB$. We will prove that the vertex $v$ cannot penetrate the tetrahedron $T=CH(B,C,D,F)$, the convex hull with vertices $B,C,D,F$, which is highlighted in Figure~\ref{10chain1} below.

\begin{figure}[htbp]
\centering
\includegraphics[width=1.0\linewidth]{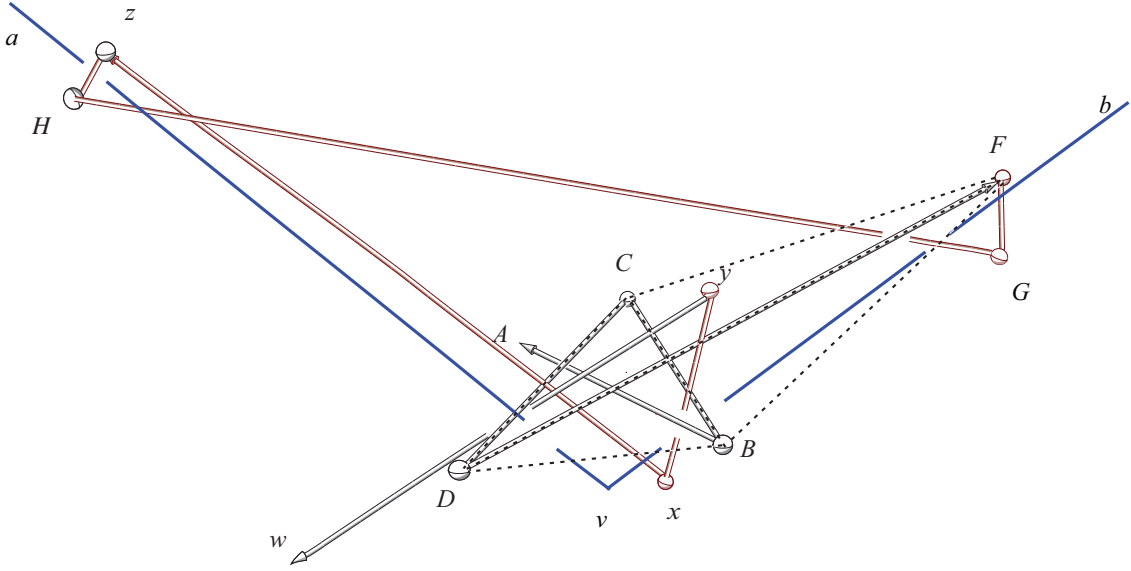}
\caption{Tetrahedron $T=CH(B,C,D,F)$ highlighted.}
\label{10chain1}
\end{figure}

Assume $v$ can penetrate $T$, that is, $v$ can move into the inside of the tetrahedron $T=CH(B,C,D,F)$. Note that the link $vb$ passes through the $\e$-ball at $F$. When $\e$ is sufficiently small, $vb$ and $DF$ come very close to $F$. We may assume $vb$ and $DF$ meet near a point $M$ on $DF$; then all points on $vM$ except $M$ stay inside $T$ as $T$ is convex. This contradicts the fact that $vb$ is weaved under $BC$ and $AB$ and the 2-chain $(a,v,b)$ straddles $BC$ and $AB$. So the vertex $v$ cannot penetrate the tetrahedron $T=CH(B,C,D,F)$.

\end{proof}
Therefore the $2$-chain links are trapped by the 3/4-tangle.
Therefore the only way the 2-chain could slide free of the near rigid triangular
frame is if one of the end vertices enters the $\e$-ball at the
jag loop corners.

\subsubsection{The Vertex $v$ Cannot Move Far}
Thus
the 2-chain $(a,v,b)$ cannot slide free of any of the triangle corners if we make the 2-chain links extra long, so that the two end points cannot enter the $\e$-ball containing the corners.
We recall Lemma 4 from \cite{GLOSZ}.

\noindent {\it When $\e$ is sufficiently small, a line piercing two disks of
radius $\e$ can angularly deviate from the line connecting the
disk centers at most $\d \le  \frac{2\e}{ m}$, where $m$ is the distance
between the disk centers.}

 Figure~\ref{Triangle.Lemma}(a) illustrates the largest angle $\d$.

\begin{figure}[htbp]
\centering
\includegraphics[width=0.8\linewidth]{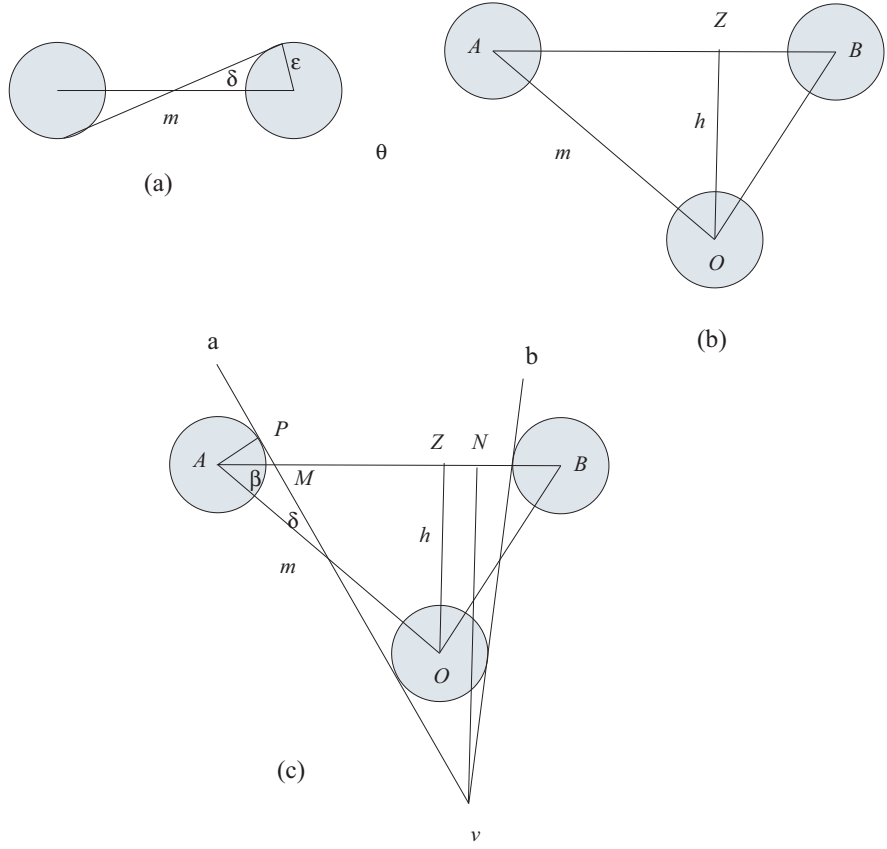}
\caption{Triangle Lemma: (a) Line through two disks deviates at
most $\d$; (b) Triangle structure; (c) The largest distance between $v$ and $AB$.}
\label{Triangle.Lemma}
\end{figure}

Let $O, A,B$ be the centers of the three $\e$-balls containing the triangle corners, see Figure~\ref{Triangle.Lemma}(b). Let the side length $|OA|=m$, and $\angle OAB=\b$, and the altitude $|OZ|=h$. Without loss of generality, we assume that $\b$ is an acute angle (otherwise $\angle OBA$ is acute, and use $\angle OBA$). Furthermore, when $\e$ is sufficiently small, $\b+ \frac{2\e}{m}$ can be assumed to remain an acute angle. Draw a perpendicular line from $A$ to line $va$ and  let $P$ be the point of intersection. Similarly, draw a perpendicular line from $v$ to line $AB$ and  let $N$ be the point of intersection. Let $M$ be the point of intersection of $AB$ and $va$.
The following lemma captures the key constraint on the motion of the 2-link chain.

\begin{lemma}
If the sides of the triangle pass through the $\e$-disks
illustrated, then the distance $|vN|$ satisfies $h-\e\leq |vN|\leq |AB|\tan(\b+\frac{2\e}{m})$. Hence $|vN|$ is bounded.
\end{lemma}
\begin{proof}
Since $va$ and $vb$ both pass through the $\e$-ball centered at $O$, the distance from $v$ to $AB$ is at least $|OZ|-\e=h-\e$. On the other hand,
Figure~\ref{Triangle.Lemma}(c) illustrates the largest distance $|vN|$ between $v$ and $AB$. Note  $\triangle vNM$ is a right triangle, $\angle vMN=\angle OAB+\d=\b+\d$, and $\tan(\b+\d)=\frac{|vN|}{|MN|}$, so
$|vN|=|MN|\tan(\b+\d)=(|AN|-|AM|)\tan(\b+\d)\leq |AN|\tan(\b+\d)\leq |AN|\tan(\b+\frac{2\e}{m})\leq |AB|\tan(\b+\frac{2\e}{m})$ since the tangent function is increasing on the interval $(0, \frac{\pi}{2})$ and $|AN|\leq|AB|$.

\end{proof}
From the proof above, $h-\e\leq |vN|\leq |AN|\tan(\b+\frac{2\e}{m})$.
Since $\lim_{\e\to 0}h-\e=h$ and $\lim_{\e\to 0}|AN|\tan(\b+\frac{2\e}{m})=|AN|\tan(\b)=h$, we obtain $\lim_{\e\to 0}|vN|=h$.

\subsubsection{Main Theorem}
We connect 3D to 2D via the plane determined by the 2-link chain in the proof
of the main theorem below.

\begin{theorem}
The 2-link chain is interlocked with the 10-link triangle chain.
\label{main.theorem1}
\end{theorem}
\begin{proof}
Let $H$ be the plane containing the 2-link chain. The
links of the 2-chain pass through $\e$-balls around the three
vertices of the triangle.

$H$ meets these balls in disks each of
radius $\le \e$. The lemma above shows that the distance $|vN|$ from the vertex $v$ to $AB$ is bounded above and below. Thus, by choosing
$\e$ small enough, we limit to any desired amount the distance vertex $v$ of the
2-link chain can be pushed toward or separated from the triangle.

Similar to \cite{GLOSZ}, we can establish that the 2-chain links are interlocked with the 3/4-tangle and jag loops through which they pass, under the assumption that the triangle is nearly rigid.

Thus, when we  choose $\e$ small enough (and $va$, $vb$ long enough) to prevent the
two end vertices of the 2-link chain from entering the $\e$-balls, it ensures
that the 2-link chain is interlocked with the triangle chain.
\end{proof}

Below we sketch an argument that suggests, without formally establishing, that 10 links is the minimum
needed for interlocking.
Our main theorem above proves that the minimum $k\leq 10$. We argue that the minimum $k$ cannot be less than 10 by contradiction. Assume there is a flexible $n$-chain in a nearly rigid triangular frame that can interlock with a flexible 2-chain and $n\leq 9$. Furthermore, assume that at two vertices of the triangle frame, jag loops are used to allow maximal sharing of links to result in a minimum number of links. Now take away the two extra links at the jag loops and the side of the triangle connecting the jag loops, we consider the remaining links in the $n$-chain, which consists of two open chains with a total of links $n-3\leq 6$. In the following we argue that two open chains with a total of $n-3\leq 6$ links either cannot occur, or cannot form a loop to trap the 2-chain at the vertex $v$. The cases of two open chains with a total of links $6 = n-3$ are:
\begin{enumerate}
\item A 1-chain and an open 5-chain;
\item A 2-chain and a 4-chain; and
\item Two 3-chains.
\end{enumerate}
Case (1) cannot happen. For, if it does occur, then the 1-chain would be used as both the end link of the open 9-chain and as a side of the triangular frame, which allows the frame some flexibility and renders it too flexible. In Cases (2) and (3), the chains cannot interlock, so they cannot interlock the 2-chain.

The cases of $n-3<6$ may be argued analogously to the cases of $n-3=6$.

The reason this sketch is not a formal proof is that we are assuming that, in order to reduce the number of links used, interlocking must be achieved
by ``rigidifying" the 2-chain to a nearly rigid triangle.
A more direct contradiction from the assumption of locking with a 9-chain would be desirable.

\end{document}